\newcounter{myctr}

\documentclass{ws-acs}

\usepackage[compress]{cite}
\usepackage{xcolor}
\usepackage[verbose]{hyperref}
\hypersetup{colorlinks=false,allbordercolors=blue,pdfborderstyle={/S/U/W 1}}
\usepackage{amsmath}
\usepackage{mathtools}
\usepackage{subcaption}
\usepackage{float}
\usepackage{booktabs}
\usepackage{multirow}
\usepackage{graphicx}
\usepackage{xcolor}

\begin{document}

\makeatletter
\def\@biblabel#1{[#1]}
\makeatother

\markboth{Chen Zhang and Jürgen Hackl}{Higher-Order Network Framework for Capturing Memory-Driven Mobility Dynamics}

%
\catchline{}{}{}{}{}
%

\title{INSTRUCTIONS FOR TYPESETTING MANUSCRIPTS\\
USING \LaTeX\footnote{For the title, try not to
use more than 3 lines. Typeset the title in 10 pt
Times Roman, uppercase and boldface.}
}
\title{Beyond Connectivity: Higher-Order Network Framework for Capturing Memory-Driven Mobility Dynamics}

\author{CHEN ZHANG}

\address{Department of Civil and Environmental Engineering, Princeton University, 59 Olden Street\\
Princeton, New Jersey 08544, United States\\
\email{cz2687@princeton.edu}}

\author{J\"URGEN HACKL\footnote{Corresponding author. }}

\address{Department of Civil and Environmental Engineering, Princeton University, 59 Olden Street\\
Princeton, New Jersey 08544, United States\\
\email{hackl@princeton.edu}}

\maketitle

\begin{history}
\received{(Day Month Year)}
\revised{(Day Month Year)}
\accepted{(Day Month Year)}
\published{(Day Month Year)}
\comby{(xxxxxxxxxx)}
\end{history}

\begin{abstract}
Understanding and predicting mobility dynamics in transportation networks is critical for infrastructure planning, resilience analysis, and traffic management. Traditional graph-based models typically assume memoryless movement, limiting their ability to capture sequential dependencies inherent in real-world mobility patterns. In this study, we introduce a novel higher-order network framework for modeling memory-dependent dynamics in transportation systems. By extending classical graph representations through higher-order Markov chains and de Bruijn graph structures, our framework encodes the spatial and temporal ordering of traversed paths, enabling the analysis of structurally and functionally critical components with improved fidelity. We generalize key network analytics, including betweenness centrality, PageRank, and next-step prediction, to this higher-order setting and validate our approach on the Sioux Falls transportation network using agent-based trajectory data generated with MATSim. Experimental results demonstrate that higher-order models outperform first-order baselines across multiple tasks, with the third-order model achieving an optimal balance between predictive accuracy and model complexity. These findings highlight the importance of incorporating memory effects into network-based transportation analysis and offer a scalable, data-driven methodology for capturing complex mobility behaviors in infrastructure systems.
\end{abstract}

\keywords{Higher-Order Networks, Transportation Systems, Markov Models, Network Analytics, Complex Infrastructure Systems.}

\newpage
\section{Introduction}
%
%
Transportation systems are critical infrastructures, which reliability, robustness, and efficient operation greatly affects national economics, politics and people's everyday life \cite{Wang2010Generating}.
Therefore, analyzing the structure and function of transportation networks has become an important field in science.
%
%
The foundation of transportation studies can be found in the field of graph theory.
%
%
This mathematical sub domain deals with the topology of networks, which are composed out of \textit{nodes} (or \textit{vertices}) connected by \textit{edges} (or \textit{links}).
%
%
Due to the limitation of data and computational power, studies in the early 1960s were focused mainly on simple topological and geometric properties \cite{Garrison1960CONNECTIVITY,Garrison1962structure,Kansky1963Structure,Haggett1969Network}.
%
%
In the following decades, the study of complex networks emerged as a unifying framework, with researchers from diverse fields analyzing the structure and dynamics of large-scale networks across a wide range of domains \cite{Watts1998Collective,Barabasi1999Emergence,Barabasi2003ScaleFree,Barabasi2009Linked,Newman2003Structure}.
%
%
Building on these advancements, transportation science leverages empirical data and computational methods to analyze heterogeneous, time-varying \cite{Cascetta1989stochastic}, and spatially embedded networks \cite{Hackl2017Generation}, enabling the study of both static topological properties (e.g., betweenness centralities \cite{Wang2008Betweenness}) and dynamic processes (e.g., traffic flow \cite{Hackl2019Estimation}).
%
%
Moreover, multilayer network models \cite{Kivela2014Multilayer} have emerged as powerful tools for analyzing the interdependencies between various transport modes (e.g., buses, trams, and walking paths), enabling simulations of disruptions or improvements in service\cite{Aleta2017Multilayer, Hackl2019Modelling}.
%
%
Together, these advancements render network science indispensable for addressing contemporary challenges in urban transportation planning, including multimodal integration, reliability, efficiency, and sustainability.

%
In the analysis of transportation networks, two fundamental questions arise: (1) which parts of the system are critical to its structure and function?, and (2) how do agents (e.g., people and vehicles) move through the network?
%
%
The first question is critical for network design and is addressed by analyzing static properties such as betweenness centrality \cite{Freeman1977Set}, closeness centrality \cite{Bavelas1950Communication}, and PageRank \cite{Page1998PageRank}, to identify critical nodes or edges within the network.
%
%
For instance, betweenness centrality is used as a static predictor of congestion and load on networks \cite{Kirkley2018betweenness}, closeness centrality provides insights into the spatial organization and accessibility of urban areas \cite{Crucitti2006Centralitya}, and PageRank enables the ranking of network nodes by incorporating both topological structure and external factors such as public services or commercial activities \cite{Agryzkov2012algorithm}.
Furthermore, network based importance measures are used to study resilience of transportation systems \cite{Hackl2018Determination, Liu2022Prioritizing}.
%
%
The second question, focuses on traffic flow optimization and is addressed by modeling network dynamics, using a large variety of methods developed in nonequilibrium statistical physics and nonlinear dynamics \cite{Mahnke2005Probabilistic}. 
%
%
Reaching from large-scale \textit{macroscopic} models \cite{Wang2022Macroscopic}, where the traffic movement resembles a compressible fluid, towards \textit{microscopic} models, where the focus shifts to the interactions between individuals \cite{Nguyen2021overview}. 
%
%
Due to their computational complexity, these classical models are increasingly complemented by network science approaches that offer scalable methods for analyzing traffic dynamics.
%
%
For instance, dynamic models, such as random walks\cite{Pearson1905Problem} formulated as Markov chains in graphs\cite{Besenczi2021Largescale} or models based on shortest paths \cite{Akbarzadeh2018Communicability}, are used to capture the probabilistic flow of agents through the network.
%
%
By simulating the stochastic nature of traffic, such models provide critical insights into congestion patterns, route choice behavior, and overall traffic flows with reduced computational cost \cite{Hackl2019Estimation}.
%
%
By combining insights from both static and dynamic network analyses, transportation engineers can more effectively design and optimize transportation systems.

%
%
In recent years, the rise of big data has significantly transformed transportation science, enabling researchers to move beyond small-scale models and limited empirical validation \cite{Yuan2021Survey,Haghani2023Emerging,Fan2025Advancements}.
%
%
The availability of high-resolution spatio-temporal data, such as GPS trajectories, mobile phone records, and sensor-based traffic measurements, provides unprecedented insights into mobility patterns \cite{Gonzalez2008Understanding,Chang2021Mobility}, infrastructure usage \cite{Moghtadernejad2022Datadriven}, and traveler behavior \cite{Wang2021Role}.
%
%
However, studies have shown that sequential data on graphs, such as mobility trajectories, which represent chronologically and topologically ordered sequences of nodes traversed by users, cannot be adequately modeled using simple graph structures that capture only dyadic relationships between pairs of entities \cite{Rosvall2014Memory,Lambiotte2019networks,Battiston2021physics}.
%
%
Thereby the presence of such ``higher-order'' interactions can fundamentally influence our understanding of complex systems \cite{Eliassi-Rad2021HigherOrder}.
They can fundamentally alter our understanding of node importance as measured by centrality metrics, impact the detection of clusters and community structures within graphs, and non-trivially influence dynamical processes such as diffusion and epidemic spreading, along with the effectiveness of associated control strategies \cite{Scholtes2014Causalitydriven,Salnikov2016Using,Xu2016Representing,Benson2016Higherorder}.
%
%
To address this challenge, researchers in topological data analysis, network science, and machine learning recently started to generalise network analysis to higher-order network models that capture more than dyadic relations \cite{Eliassi-Rad2021HigherOrder}. 
%
%
In transportation science, the integration of higher-order modeling frameworks that account for \textit{memory effects}, i.e. capturing the spatial and temporal ordering of interactions observed in real-world mobility data, presents a promising direction to address the inherent limitations of traditional, first-order network representations \cite{Perri2021Applications}.
%
%
By explicitly incorporating the sequential dependencies observed in empirical mobility trajectories, these models could enable a more accurate identification of structurally and functionally critical components, while offering improved representations of dynamic processes such as routing behavior and traffic flow predictions.

%
In this work, a novel higher-order network framework is presented for modeling memory-dependent dynamics in transportation systems.
This approach extends traditional graph-based representations by incorporating sequential dependencies observed in empirical mobility trajectories through higher-order Markov chains and a de Bruijn-based network structure.
In this representation, nodes correspond to subpaths of traversed road segments rather than individual intersections, allowing the encoding of both topological and temporal information within a unified graph structure.
By integrating this higher-order representation with empirically derived transition probabilities, the framework captures complex routing behavior and provides more accurate approximations of dynamic processes such as traffic flow and agent movement.
This enriched network representation enables improved predictive performance, enhances centrality-based structural analyses, and offers new insights for understanding the functional behavior of transportation networks under real-world conditions.
Specifically, this work advances the state-of-the-art in the field  of complex network science in transportation research as follows.
\begin{itemize}
\item Using a higher-order network representation for transportation data, it becomes formally feasible to account for memory effects inherent in real-world mobility trajectories. In this representation, nodes correspond to subpaths of length $k$ rather than single intersections, and edges encode transition probabilities between these memory-dependent subpaths.
\item The proposed framework extends classical network measures, such as betweenness centrality, PageRank, and next-step prediction, into higher-order variants that incorporate memory-dependent transitions, yielding a more accurate representation of flow patterns, node importance, and agent dynamics compared to traditional first-order approaches.
\item The integration of higher-order network models enables a data-driven determination of the optimal memory length and dual-scale analysis, as validated on the Sioux Falls benchmark network, where higher-order representations outperform conventional first-order models in accuracy, predictive power, and robustness.
\end{itemize}

%
The remainder of this work is structured as follows. Section~\ref{sec:methodology} introduces the mathematical foundations of higher-order network modeling, including the formalization of memory effects in transportation systems and the construction of de Bruijn-based network representations. Section~\ref{sec:analytics} extends classical network analytic techniques, such as betweenness centrality, PageRank, and next-step prediction, to higher-order models, enabling a more accurate characterization of structural importance and dynamic processes. Section~\ref{sec:application} presents an empirical application of the proposed framework to the Sioux Falls transportation network using agent-based trajectory data generated via MATSim. Section~\ref{sec:results} reports on the comparative performance of first- and higher-order models across analytical tasks, while Section~\ref{sec:discussion} discusses the methodological implications, generalizability, and limitations of higher-order representations in transportation science. Finally, Section~\ref{sec:conclusions} summarizes key findings and outlines directions for future research.

\newpage
\section{Transportation Systems as Higher-Order Networks}\label{sec:methodology}

Transportation systems are naturally represented as networks, where intersections and road segments correspond to nodes and edges, respectively.
Classical graph-based models provide a foundational framework for analyzing such systems, enabling the study of structural properties and traffic dynamics.
However, traditional models often assume memoryless movement, which fails to capture the sequential dependencies observed in real-world mobility data.
This section introduces the theoretical foundations of higher-order network representations that extend classical graph models by incorporating memory effects through higher-order Markov chains and de Bruijn graph structures.
These enriched representations provide a formal basis for modeling temporally ordered movement patterns and serve as the foundation for the analytical methods developed in subsequent sections.

\subsection{Transportation Systems as Networks}
%
%
Formally, a \textit{graph} or \textit{network} \(G=(\mathcal{V},\mathcal{E})\) is a mathematical structure, which describes a set of nodes (or vertices) \(\mathcal{V}\) and a set of edges (or links)  \( \mathcal{E} \subseteq \mathcal{V} \times \mathcal{V} \).
The nodes correspond to the objects that are modeled (e.g. crossings), the edges indicate some relation between these objects (e.g. roads).
In this work, networks are assumed to be \textit{directed}. Therefore, an edge is an ordered pair \((v,w)\in E\) indicating that \(v\) and \(w\) are directly connected and transport takes place only from \(v\) to \(w\).
The reversed assertion is not possible, unless an edge \((w,v)\) exists. In the context of transportation networks, this directed structure reflects lane-level directionality of the system.

%
A \textit{path} \(p\) in a network is defined as a sequence of nodes, \(p:= (v_1, v_2, \ldots, v_{\ell}) \; \forall v_{i} \in \mathcal{V} \) such that every consecutive pair of nodes is connected by an edge \((v_i, v_{i+1}) \in \mathcal{E}\) \(\forall i \in \{1,\dots,\ell\}\).
When no node is repeated in the sequence (i.e., \(v_i \neq v_j\) \(\forall i,j: 1 \leq i < j \leq \ell\)), the path is considered to be \textit{simple}.
The \textit{distance} \(\mathrm{d}(v,w)\) is the minimum length of any path between nodes \(v\) and \(w\), with $\mathrm{d}(v,w) := \infty \Leftrightarrow \nexists $ path from $v$ to $w$.
Among all possible paths, the \textit{shortest paths} is the path between two nodes  \(v\) and \(w\), such that no shorter path exists.
In weighted networks a path \(p=(v_1,\ldots,v_{\ell})\)  is referred to as the \textit{cheapest path} if it minimizes the total weight, i.e. \(\min_{p}\sum_{i=1}^{\ell} \mathrm{w}(v_{i}, v_{i+1})\).

%
While the computation of cheapest paths is fundamental in traffic flow modeling, identifying such paths between all pairs of nodes under dynamically varying edge weights poses significant computational challenges, particularly in large-scale networks \cite{Hackl2019Estimation}.
As an alternative, network dynamics can be modeled as stochastic processes, which offer scalable and probabilistic representations of agent movement and flow \cite{Masuda2017Random}.
Among these, random walks \cite{Pearson1905Problem} constitute a widely used class of models for capturing such stochastic dynamics on networks.

%
Formally, a random walk is defined as a stochastic process on a state space \(\Omega\).
Let \(X_t \in \Omega\) denote the state of the walker at time \(t\).
The sequence \(\{X_t\}_{t \ge 0}\) forms a random walk, with both the time parameter \(t\) and the state space \(\Omega\) being either discrete or continuous.
A particularly simple class of random walk models are based on first-order Markov chains and assume a \textit{memoryless} process, meaning the probability of moving to the next state \(s_{{t+1}}\) depends solely on the current state \(s_{{t}}\), expressed as \(P ( X_{t+1} = s_{t+1} \mid X_t=s_{t}, X_{t-1}=s_{t-1}, \ldots, X_0=s_0) = P ( X_{t+1}=s_{t+1} \mid X_t=s_{t} ) \), for all sequences of states \(s_0, \ldots, s_{t-1} \in \Omega\) and \(t \ge 0\).
For a random walk on a graph \cite{Lovasz1994Random}  \(G = (\mathcal{V}, \mathcal{E})\), where \(\Omega = \mathcal{V}\), the transition probability \(P (X_{t+1} = v_{j} \mid X_{t}=v_{i} )\) is given by the transition matrix \(\bm{T}:= \bm{A}_{ij}/\sum_{l} \bm{A}_{il}\) for each entry \(ij\) of the adjacency matrix \(\bm{A}\).

\subsection{Memory Effects in Transportation Systems}
%
%
Empirical analyses of transportation systems indicate the existence of \textit{memory effects} in routing decisions, whereby drivers' choices are influenced not only by their current location but also on the sequence of previously traversed nodes \cite{Zhu2015People, Lima2016Understanding, Manley2015Shortest}.
%
%
To capture these effects, higher-order Markov chains extend the conventional model by conditioning the transition probabilities on the ordered history of visited nodes.
Specifically, instead of modeling the transition probability \(P(s_{t+1} \mid s_t)\) as in first-order Markov chains, a higher-order Markov chain estimates \(P(s_{t+1} \mid s_{t}, \ldots, s_{t-k+1})\), where the order \(k\) represents the length of the memory encoded in the model \cite{Singer2014Detecting}.
Given a  \(k\)th-order Markov chain with transition probabilities \(P^{(k)}(s_t \mid s_{t-1}, \ldots s_{t-k}\), the likelihood of observing a path \( p= ( v_{0}, v_{1}, \ldots  v_{\ell})\) under a model \(M_k\) is given by \cite{Scholtes2017When}:
\begin{equation}\label{eq:1}
  \mathcal{L}\left(M_k \mid p\right) = \prod_{i=k}^{\ell} P^{(k)}\left(v_i \mid v_{i-1}, \ldots , v_{i-k} \right).
\end{equation}
When a multi-set \(\mathcal{P} = \{p_1, p_2, \ldots, p_N\}\) of statistically independent paths is observed, the overall likelihood of the \(k\)th-order Model \(M_{k}\) is given by
\begin{equation}\label{eq:2}
  \mathcal{L}\left(M_k \mid \mathcal{P}\right) = \prod_{j=1}^{N} \mathcal{L}\left(M_k \mid p_j\right).
\end{equation}
Parameter estimation is performed via maximum likelihood, where the transition probabilities are inferred from the empirical frequencies sub-paths in the data:
\begin{equation}\label{eq:3}
    \widehat{P}^{(k)}\left(v_i \mid v_{i-1}, \ldots , v_{i-k} \right) = \frac{\left|\{(v_{i-k}, \ldots, v_{i-1}, v_i) \in \mathcal{P}\}\right|}{\sum_{w \in \mathcal{V}}\left|\{(v_{i-k}, \ldots, v_{i-1},  w) \in \mathcal{P}\}\right|}.
\end{equation}
where the numerator counts the occurrences of the specific sequence \((v_{i-k},\ldots,v_{i})\), while the denominator sums oval all possible continuations \(w\in\mathcal{V}\) that share the same prefix sub-path \((v_{i-k},\ldots,v_{i-1})\), ensuring normalization of the transition probability.

%
The following proposition formally demonstrates that, under the assumption of independently and identically distributed (i.i.d.) path samples, the maximum likelihood estimator defined in Equation~\eqref{eq:3} provides an unbiased estimate of the underlying transition probabilities in a \(k\)th-order Markov model.

\begin{proposition}[Unbiasedness of the MLE for Transition Probabilities]
  Consider a multi-set \( \mathcal{P} = \{p_1, p_2, \ldots, p_N\} \) of statistically independent paths generated by a \(k\)th-order Markov process with true transition probabilities
  \begin{equation*}
    P^{(k)}\left(v_i \mid v_{i-1}, \ldots , v_{i-k} \right) = p.
  \end{equation*}
  Let \( Y \) denote the total number of occurrences of the prefix \((v_{i-k}, \ldots, v_{i-1}) \in \mathcal{P} \) (with \( Y>0 \)) and let \( X \) denote the number of times the sub-path \((v_{i-k}, \ldots, v_i)\) appears. Then, under the assumption that
  \begin{equation*}
    X \mid Y \sim \operatorname{Binomial}(Y, p),
  \end{equation*}
  the maximum likelihood estimator (MLE)
  \begin{equation*}
    \widehat{P}^{(k)}\left(v_i \mid v_{i-1}, \ldots, v_{i-k}\right) = \frac{X}{Y}
  \end{equation*}
  is an unbiased estimator of \( p \); that is,
  \begin{equation*}
    \mathbb{E}\left[\widehat{P}^{(k)}\left(v_i \mid v_{i-1}, \ldots, v_{i-k}\right)\right] = p.
  \end{equation*}
\end{proposition}

\begin{proof}
  Since \( X \mid Y \sim \operatorname{Binomial}(Y, p) \), it follows that
  \begin{equation*}
    \mathbb{E}[X \mid Y] = Y \, p.
  \end{equation*}
  Therefore, the conditional expectation of the estimator is given by
  \begin{equation*}
    \mathbb{E}\left[\frac{X}{Y} \,\bigg| \, Y\right] = \frac{\mathbb{E}[X \mid Y]}{Y} = \frac{Y \, p}{Y} = p.    
  \end{equation*}
  Taking the total expectation over \( Y \) using the law of total expectation, we have
  \begin{equation*}
  \mathbb{E}\left[\widehat{P}^{(k)}\left(v_i \mid v_{i-1}, \ldots, v_{i-k}\right)\right] = \mathbb{E}\left[\mathbb{E}\left[\frac{X}{Y} \, \bigg| \, Y\right]\right] = \mathbb{E}[p] = p.    
  \end{equation*}
\end{proof}

\subsection{Higher-Order Network Representation}
%
To formally capture the memory effects inherent in real-world mobility trajectories using higher-order Markov chains, it is necessary to encode both topological and temporal information within a unified network representation.
%
%
A second-order memory model (\(k=2\)) can be interpreted as a line graph \(L(G)\) of the original first-order network \(G = (\mathcal{V}, \mathcal{E})\), where each node in \(L(G)\) corresponds to a directed edge \((v_i, v_j) \in \mathcal{E}\) in \(G\), and an edge exists between two nodes \(((v_i, v_j), (v_j, v_k)) \in \mathcal{E}_{L(G)}\) if the corresponding transitions are path-continuous.
This transformation implies that each state in the second-order model contains information not only about the current position \(v_j\) but also about the preceding location \(v_i\).
%
%
Extending this reasoning, a third-order memory model (\(k=3\)) corresponds to the line graph of the line graph, \(L^2(G)\), where nodes now encode subpaths of the form \((v_i, v_j, v_k)\), thus capturing the current location as well as the two preceding steps.
%
%
More generally, a \(k\)th-order memory model corresponds to a \textit{de Bruijn graph} \cite{Bruijn1946combinatorial} of order \(k\), where each node represents a sequence of \(k\) consecutive nodes, and directed edges connect overlapping subpaths, encoding valid transitions observed in empirical trajectories.

%
Based on this a higher-order network model with path-dependencies can be formulated \cite{Scholtes2017When,Qarkaxhija2022Bruijn}, where for a chosen order \(k\) the higher-order nodes are defined as:
%
%
\begin{equation}
  \label{eq:4}
  v^{(k)}:= \langle v_{1},v_{2},\ldots,v_{k} \rangle  = \left\{\left(v_1, v_2, \ldots, v_k\right) \mid v_i \in \mathcal{V}\right\}
\end{equation}
where \(\left(v_1, v_2, \ldots, v_k\right)\) is a possible path in \(G\).
%
%
Thereby, higher-order edges \((v^{(k)},w^{(k)})\) capture transitions between these paths.
Let \(v^{(k)}= \langle v_{1},v_{2},\ldots,v_{k} \rangle\)  and \(w^{(k)}= \langle w_{1},w_{2},\ldots,w_{k} \rangle\) be two \(k\)th-order nodes, an higher-order edge exists if \(v^{(k)}\) and \(w^{(k)}\) are equal except for the first entry of  \(v^{(k)}\) and the last entry of \(w^{(k)}\), such that \(v^{(k)}_{i+1} = w^{(k)}_{i} \;\forall i=1,2,\ldots,k-1\).
The resulting edge is a sequence of
\begin{equation}
  \label{eq:5}
  (v^{(k)},w^{(k)}) = (\langle v_{1},v_{2},\ldots,v_{k} \rangle , \langle v_{2},\ldots,v_{k}, v_{k+1} \rangle ) = ( v_{1},v_{2},\ldots,v_{k},v_{k+1}).
\end{equation}
The weight of this edge can be defined by the conditional probability that captures the memory effect:
\begin{equation}
  \label{eq:6}
  \text{w} \left(\langle v_{1},v_{2}, \ldots , v_{k} \rangle , \langle v_{2}, \ldots, v_{k}, v_{k+1} \rangle\right):= P \left( v_{k+1} \mid v_{k}, \ldots, v_{2}, v_{1} \right).
\end{equation}

%
Two primary approaches exist for constructing higher-order networks that capture sequential memory: (1) topology-based methods, which derive higher-order structures from the graph’s connectivity, and (2) path-based methods, which infer memory-dependent transitions directly from observed trajectory data.
%
%
In the topology-based approach, a \(k\)th-order network is constructed by representing each unique sequence of \(k\) consecutive nodes, extracted from paths in the first-order network, as a node in the corresponding de Bruijn graph \cite{Compeau2011How}.
This representation encodes memory effects by explicitly modeling transitions between overlapping subpaths of length \(k\), thereby conditioning each step on the \(k-1\) preceding states.
While this approach can be used without any additional knowledge on network dynamics, it suffers from scalability limitations, as the number of nodes and edges grows significantly with \(k\), resulting in substantial memory and computational demands for large-scale networks or higher-order representations.
Recent advances have partially mitigated these challenges through the use of probabilistic data structures and minimizer-based techniques \cite{Pell2012Scaling, Rautiainen2021MBG}, yet scalability remains a critical concern in applications involving deep memory structures.
%
%
When path-level data from the system is available, a higher-order model can be directly inferred from observed trajectories. To capture the sequential context inherent in such data, a fixed-length sliding window of size \(k+1\) is applied to each trajectory, extracting contiguous node sequences of the form \((v_1, v_2, \ldots, v_k, v_{k+1})\).
Each subpath of length \(k\), such as \((v_1, v_2, \ldots, v_k)\), is treated as a distinct higher-order node, effectively encoding the memory of the preceding steps.
Connectivity between higher-order nodes is then established by identifying overlaps between successive subpaths.
Specifically, a directed edge is introduced from the node representing \((v_1, v_2, \ldots, v_k)\) to the node representing \((v_2, v_3, \ldots, v_{k+1})\), thereby encoding the transition observed in the data.
This approach enables the encoding of empirical trajectory data into a higher-order network structure in a computationally tractable manner, as it requires consideration only of observed subpaths rather than all combinatorially possible sequences.

%
A fixed $k$th-order model assume a uniform memory length across all transitions, which can impose restrictive assumptions on the underlying dynamics.
In practice, the amount of available historical context often varies along a trajectory. 
For instance, at the beginning of a trajectory, insufficient historical context may prevent the application of higher-order dependencies, necessitating the use of lower-order transition probabilities.
Furthermore, in forecasting tasks such as next-step prediction, both short-term and long-term dependencies can influence the prediction accuracy.
To address this limitations, Scholtes \cite{Scholtes2017When} proposed a multi-order framework that allows for the simultaneous consideration of multiple correlation lengths.
Moreover, this approach provides statistically based method for selecting the appropriate model order \(k\) based on the observed data, by conducting a likelihood ratio test across different model orders. Additional information is provided in \ref{app:optimal-order}.

\begin{figure}[h]
  \centering
  \includegraphics[width=0.8\textwidth]{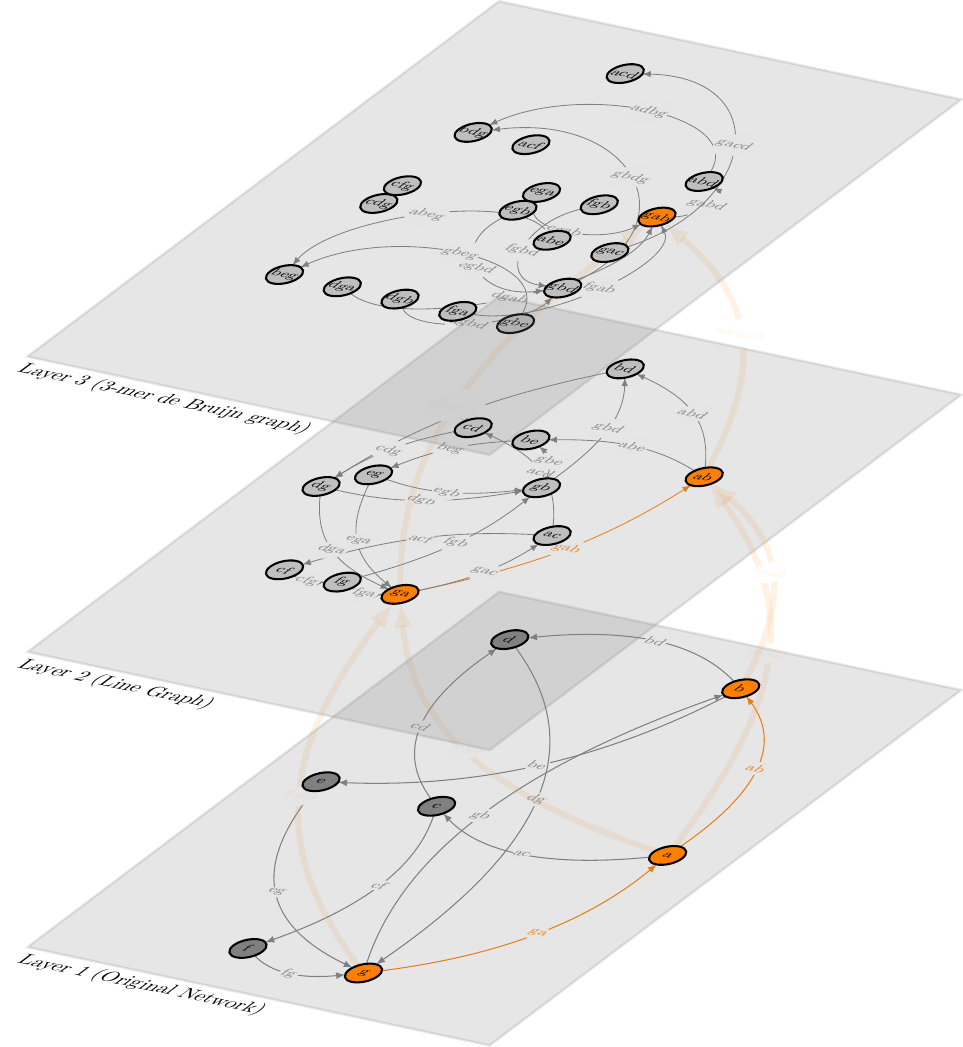}
  \caption{Hierarchical, bottom-up construction of a Higher-Order Network (HON). At layer 1 (bottom), the standard network representation is shown with vertices from $\mathcal{V}$ and their directed edges. Layer 2 (middle) is the line graph constructed from the first layer by converting each edge into a node; if two edges in the first layer are connected (i.e., share a common vertex), an edge is drawn between their corresponding nodes. For example, the edge $(g,a)$ and the edge $(a,b)$ become the line graph nodes $\langle g,a \rangle$ and $\langle a,b \rangle$, respectively. The third layer is a 3-mer De Bruijn graph, where higher order nodes such as $\langle g,a,b \rangle$ are formed by aggregating connected nodes from the line graph. Interlayer edges illustrate how nodes from the original network are successively aggregated to construct the higher order representations.
  }
  \label{fig:hon-construction}
\end{figure}

\section{Higher-Order Network-Analytics}\label{sec:analytics}
%
%
%
Building upon the higher-order representation of transportation systems, this section introduces analytical methods that generalize classical network metrics and predictive models to account for memory-dependent dynamics.
By incorporating sequential dependencies into the analysis, these methods enable a more accurate characterization of node importance, flow distributions, and agent trajectories in complex transportation networks.
Specifically, we focus on three key extensions: (1) higher-order betweenness centrality, (2) higher-order PageRank, and (3) next-step prediction based on higher-order transition structures.

\subsection{Higher-Order Betweenness Centrality}
%
%
Betweenness centrality is a classical network metric that quantifies the importance of a node based on its role in facilitating shortest-path flows between other pairs of nodes.
%
%
Formally, the betweenness centrality \(C_B(v)\) of a node \(v \in \mathcal{V}\) in a graph \(G = (\mathcal{V}, \mathcal{E})\) is defined as \cite{Freeman1977Set}: \(C_B(v)=\sum_{s \neq v \neq t} \frac{\sigma_{st}(v)}{\sigma_{st}} \;\forall s, t \in \mathcal{V}\), where \(\sigma_{st}\) denotes the total number of shortest paths between nodes \(s\) and \(t\), and \(\sigma_{st}(v)\) is the number of those paths that pass through node \(v\).
%
%
In the context of transportation systems, betweenness centrality serves as a static proxy for identifying structurally critical intersections or road segments that are likely to experience high volumes of traffic or congestion \cite{Wu2022spatial, Xiang2023Estimation}. Nodes with high betweenness centrality often represent potential bottlenecks or hubs of flow, and their identification is instrumental for tasks such as network design, vulnerability assessment, and resilience analysis \cite{Kirkley2018betweenness}.
%
%
A major limitation of classical betweenness centrality is that it considers all shortest paths between all node pairs, regardless of actual origin-destination demand, and assumes that all agents follow strictly shortest-path routing, which does not reflect the heterogeneity and variability of real-world travel behavior \cite{Zhu2015People,Williams2022shape}.

%
To overcome the limitations of classical betweenness centrality, we extend the concept to higher-order networks, allowing for the incorporation of memory-dependent routing behavior.
Analogous to its first-order counterpart, the higher-order betweenness centrality of a node \(v \in \mathcal{V}\) is defined as
\begin{equation}
  \label{eq:7}
  C_B^{(k)}(v) = \sum_{s \neq v \neq t} \frac{\sigma_{st}^{(k)}(v)}{\sigma_{st}^{(k)}} \quad \forall s,t,v\in \mathcal{V},
\end{equation}
where \(\sigma_{st}^{(k)}\) denotes the total number of shortest paths between origin \(s\) and destination \(t\) in the \(k\)th-order network, and \(\sigma_{st}^{(k)}(v)\) counts how often the first-order node \(v\) appears in these higher-order paths.
This formulation accounts for the sequential context encoded in higher-order paths while still attributing centrality scores to the original nodes in the transportation network.

\subsection{Higher-Order PageRank}
%
%
PageRank is a network centrality measure that evaluates the relative importance of nodes based on the structure of incoming links and the importance of the linking nodes.
Originally developed for ranking web pages, it is defined as the stationary distribution of a random walk with teleportation on a directed graph. Formally, the PageRank vector \(\bm{r}\) satisfies \cite{Page1998PageRank}:
\(
\bm{r} = \bm{r} \left(\alpha \, \bm{T} + (1 - \alpha) \frac{1}{n} \bm{E}  \right)
\),
where \(\bm{T}\) is the row-stochastic transition matrix derived from the adjacency matrix \(\bm{A}\), \(\alpha \in (0,1)\) is the teleportation (or damping) factor, and \(\bm{E}\) matrix of all ones, such that \(\frac{1}{n}\bm{E}\) ensuring ergodicity.
In transportation systems, PageRank can be used to identify important locations (e.g., intersections, transit stations) not only based on direct connectivity but also by accounting for the recursive influence of neighboring nodes \cite{Jiang2009Ranking}.
Unlike betweenness centrality, which emphasizes shortest paths, PageRank captures the long-term visitation frequency of a random walker and is particularly useful for evaluating accessibility, planning transit-oriented development, and prioritizing infrastructure investments \cite{Agryzkov2012algorithm}.

%
While the classical PageRank algorithm assumes memoryless dynamics, it can be extended to a higher-order network representation by computing the higher-order PageRank vector \(\bm{r}^{(k)}\) as the stationary distribution of a random walk on the higher-order graph.
Specifically, the vector \(\bm{r}^{(k)} \in \mathbb{R}^m\) satisfies the equation:
\begin{equation}
  \label{eq:8}
  \bm{r}^{(k)} = \bm{r}^{(k)} \left( \alpha \, \bm{T}^{(k)} + (1 - \alpha) \frac{1}{m} \bm{E} \right),
\end{equation}
where \(\bm{T}^{(k)}\) is the row-stochastic transition matrix of the higher-order network \(G^{(k)}\).
%
%
To obtain a ranking over the original first-order nodes, the higher-order PageRank scores are projected by mapping each higher-order node to its last component, which corresponds to the walker's current position in the encoded trajectory.
Formally, let \(\mathcal{V}^{(k)}\) denote the set of higher-order nodes, where each node \(v^{(k)} \in \mathcal{V}^{(k)}\) represents a path \(p = (v_{1} , v_{2} , \ldots , v_{k})\).
The aggregated first-order PageRank score \(\tilde{r}(v)\) for a node \(v \in \mathcal{V}\) is computed as:
\begin{equation}
  \label{eq:9}
  \tilde{r}(v) = \sum_{\substack{v^{(k)} \in \mathcal{V}^{(k)} \\ \text{last}(v^{(k)}) = v}} r^{(k)}(v^{(k)}),
\end{equation}
where \(r^{(k)}(v^{(k)})\) is the PageRank score of the higher-order node \(v^{(k)}\), and \(\text{last}(v^{(k)}) = v_{k}\) refers to the final element in the subpath represented by \(v^{(k)}\).
This aggregation reflects the stationary probability of the walker arriving at node \(v\) via all memory-dependent paths, thereby incorporating sequential dependencies into the node ranking.

\subsection{Next-Step Prediction}
%
%
Beyond structural analysis, higher-order network models also enable more accurate predictive tasks, such as next-step prediction, where the goal is to estimate the most likely subsequent node in a trajectory given its recent history.
For an individual path \( p = \left(v_1, v_2 , \ldots v_\ell\right) \), the likelihood under a multi-order Markov model with maximum order \(K\) is given by
\begin{equation}
  \label{eq:10}
\bar{P}^{(K)}(p) = \prod_{k<K} P^{(k)}\left(v_k \mid v_{k-1}, \ldots, v_1 \right) \times \prod_{K \leq i \leq \ell} P^{(K)}\left(v_i \mid v_{i-1}, \ldots, v_{i-K}\right),  
\end{equation}
where \(P^{(k)}\) denotes the transition probability under a \(k\)th-order Markov model.
The first product accounts for the initial steps of the path, where limited history requires the use of lower-order models (\(k < K\)).
Once a full context of length \(K\) is available, the second product applies the \(K\)th-order transition probabilities for the remainder of the sequence.
Given a collection of observed paths \(\mathcal{P}\), we can compute the posterior probability of a candidate path \(p\) using Bayes' theorem:
\begin{equation}
  \label{eq:11}
  P(p \mid \mathcal{P}) \propto \bar{P}^{(K)}(p) \cdot P(p),
\end{equation}
where \(\bar{P}^{(K)}(p)\) is the likelihood under the multi-order model, and \(P(p)\) is a prior distribution reflecting domain-specific knowledge, such as known routing constraints, behavioral tendencies, or infrastructure limitations.

Specifically, our next-step prediction is performed by computing the conditional transition probability \(P(v_i \mid v_{i-1}, \ldots, v_{i-K})\) for the next node in the sequence, and then selecting the node with the highest probability; while the prior \(P(p)\) can be incorporated to reflect domain-specific constraints, in our implementation we rely primarily on the learned transition probabilities derived from observed path frequencies.
Future work could leverage the non-attributed network as the prior \(P(p)\) to guide predictions, potentially incorporating domain-specific routing constraints to further enhance prediction accuracy; detail will be discussed in the Discussion section.

\begin{figure}[htb]
  \centering
  \includegraphics[width=0.5\linewidth]{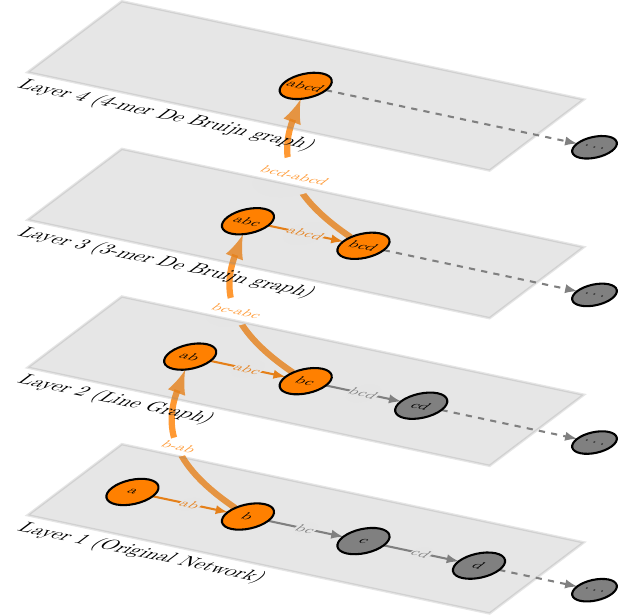}
  \caption{Aggregation Multi-Order Network that visualizes a random walk process across four layers. In Layer 1 (Original Network), the random walk begins by moving from node $a$ to node $b$ along the edge $(a,b)$. Following this step, the process transitions upward: the edge $(a,b)$ is aggregated into the node $\langle a,b \rangle$ in Layer 2 (Line Graph), where the random walk continues. Similarly, in Layer 3 (3-mer De Bruijn graph), higher order nodes (e.g., $\langle a,b,c \rangle$) are formed by aggregating connected nodes from the previous layer, and the random walk proceeds on this level. Finally, in Layer 4 (4-mer De Bruijn graph), the process further aggregates to form nodes like $\langle a,b,c,d \rangle$. Dashed inter-layer edges indicate the transition of the random walk from one layer to the next, while gray dummy nodes at the end of each layer denote that the network continues beyond the displayed segment.
  }
  \label{fig:multi-order}
\end{figure}

\newpage
\section{Application: Sioux Falls Road Networks}\label{sec:application}
To demonstrate the practical applicability and effectiveness of the proposed higher-order network framework, we apply our methodology to the well-known Sioux Falls road network \cite{Leblanc1975Algorithm}, a benchmark dataset widely used in transportation research for evaluating routing algorithms and traffic flow models \cite{Stabler2025Transportation}.
Using synthetic trajectory data generated via agent-based simulations, we construct both first-order and higher-order network models and evaluate their ability to capture structural and dynamic properties of the system.
This application enables a systematic comparison of model performance across multiple analytical tasks, including centrality estimation and next-step prediction, under controlled and reproducible conditions.

\begin{figure}[h]
    \centering
    \begin{subfigure}[b]{0.47\textwidth}
        \centering
        \includegraphics[width=\textwidth]{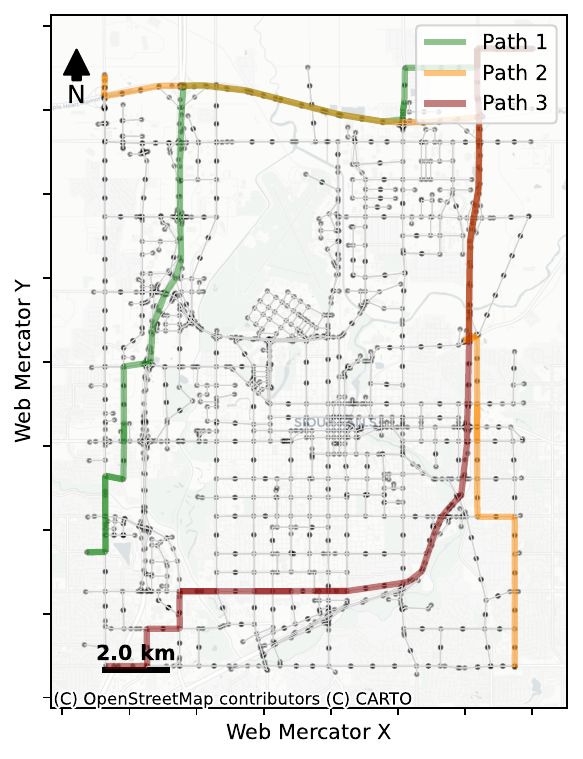}
        \caption{Network visualization of Sioux Falls network.}
        \label{fig:sf-network}
    \end{subfigure}
    \hfill
    \begin{subfigure}[b]{0.51\textwidth}
        \centering
        \includegraphics[width=\textwidth]{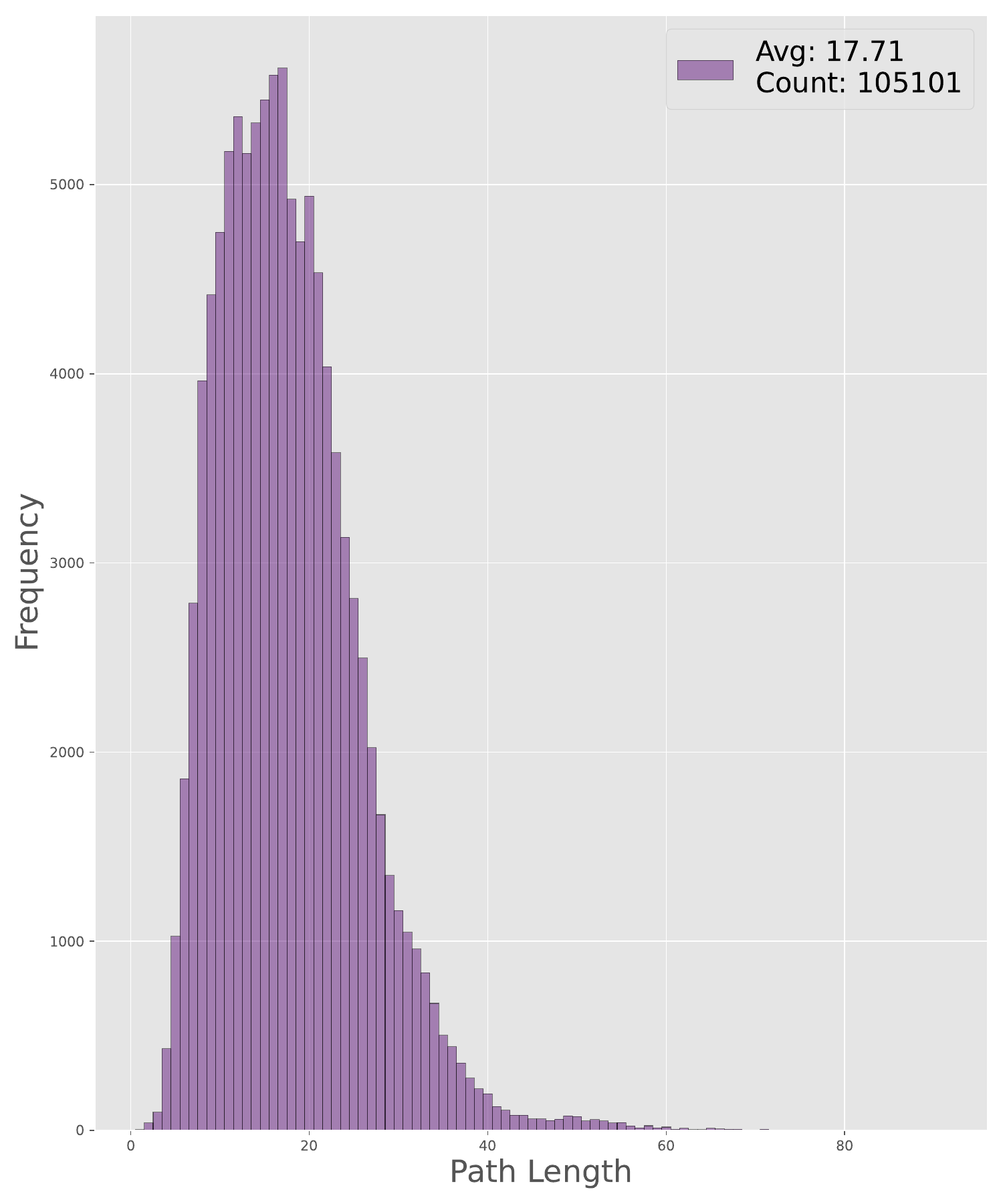}
        \caption{Distribution of path lengths for Sioux Falls road networks.}
        \label{fig:sf-paths}
    \end{subfigure}
    \caption{Overview of the application example. (a) Visualization of the Sioux Falls road network where intersections and road segments form the underlying topology. Three distinct paths are highlighted to illustrate example travel routes. (b) Distribution of path lengths computed from 105101 paths with an average length of 17.71. The histogram shows that most paths are relatively short, with a sharp drop-off in frequency for longer paths and a long tail extending toward path lengths of about 80. 
}
    \label{fig:path_length_and_network}
\end{figure}

\subsection{Trajectory Data}
%
%
We evaluate the performance of the proposed higher-order network framework under controlled and reproducible conditions by generating synthetic trajectory data through agent-based simulations on the Sioux Falls network.
%
%
Despite being synthetic, this dataset is widely recognized and validated within the transportation science and engineering community \cite{Chakirov2014Enriched,Hackl2019Estimation}, making it a suitable benchmark for studying routing behavior and traffic flow dynamics while mitigating data availability and privacy concerns.
%
%
While the original Sioux Falls network comprises 24 nodes and 76 edges \cite{Leblanc1975Algorithm}, this work utilizes an enriched network representation with 1,810 nodes and 2,494 edges.
In this representation, nodes correspond to intersections and locations of interest (e.g., homes, shops, workplaces), while edges represent directed road segments.
An overview of the network topology is provided in Figure~\ref{fig:sf-network}, and key network metrics are summarized in Table~\ref{tab:sf-network-metrics}.

\begin{table}[htbp]
  \centering
  \footnotesize
  \caption{Sioux Falls Network Metrics.}
  \label{tab:sf-network-metrics}
  \begin{tabular}{lc}
    \toprule
    \textbf{Metric} & \textbf{Value} \\
    \colrule
    Number of Nodes & 1810 \\
    Number of Edges & 2494 \\
    Average Degree & 2.7558 \\
    Density & 0.001523 \\
    Diameter & 72 \\
    Average Shortest Path Length & 28.2426 \\
    Average Clustering Coefficient & 0.036501 \\
    Global Efficiency & 0.050455 \\
    \botrule
  \end{tabular}
\end{table}

%
%
The reproduction of individual daily mobility patterns is conducted using the open-source Multi-Agent Transport Simulation (MATSim) framework \cite{Horni2016MultiAgent}.
MATSim implements an iterative, agent-based methodology for dynamic traffic assignment, integrating a queue-based traffic flow model that captures time-dependent travel times, spillback effects, and modal interactions \cite{Hackl2019Epidemic}.
Agent behavior is governed by a co-evolutionary algorithm, wherein individuals iteratively adjust their travel choices, such as routes, modes, and departure times, to maximize the utility of their daily activity schedules \cite{Nagel2009Agentbased}.
These decisions are evaluated using a random utility model, and through repeated iterations, the system converges toward an agent-based Stochastic User Equilibrium \cite{Chakirov2014Enriched}.
%
%
For the Sioux Falls network, a total of 54,546 agents are simulated over the course of a single day, resulting in 105,101 recorded mobility trajectories that capture individual travel behavior and route choices.
A statistical summary of the resulting path lengths is presented in Figure~\ref{fig:sf-paths}.

\subsection{Experimental Setup}
%
%
Leveraging the synthetic trajectory data derived from the Sioux Falls network, this study addresses two central research questions:
(i) to what extent do higher-order network models enhance the representation of structural and dynamic properties in transportation systems compared to first-order models,
(ii) how does the incorporation of empirical mobility patterns affect the performance of both first- and higher-order representations.
%
%
To address these questions, we construct two baseline first-order models.
%
%
A \textit{non-attributed} model, which encodes only the unweighted topological structure of the network, and an \textit{attributed} model, which integrates empirical edge traversal frequencies as weights, thereby reflecting observed mobility patterns.
%
%
Extending this approach, four higher-order models corresponding to memory orders \(k = 2\) through \(k = 5\) are derived from the trajectory data. Each model is implemented in both non-attributed and attributed variants.
%
%
The non-attributed models retain only the structural connectivity inferred from the occurrence of consecutive node sequences, while attributed models incorporate empirical traversal frequencies as weights between higher-order nodes.

%
%
The experimental evaluation focuses on three core analytical tasks: (1) betweenness centrality, (2) PageRank, and (3) next-step prediction.
Each method is adapted to both first- and higher-order network representations to assess their effectiveness in capturing structural importance, flow dynamics, and predictive accuracy.
For the betweenness centrality analysis, edge weights are defined as the negative logarithm of conditional transition probabilities, assigning lower costs to more frequently traversed paths.
In contrast, the PageRank and next-step prediction analyses directly use the conditional transition probabilities as edge weights, reflecting the empirically derived likelihood of movement between nodes or subpaths.
\begin{figure}[htbp]
    \centering
    \includegraphics[width=\linewidth]{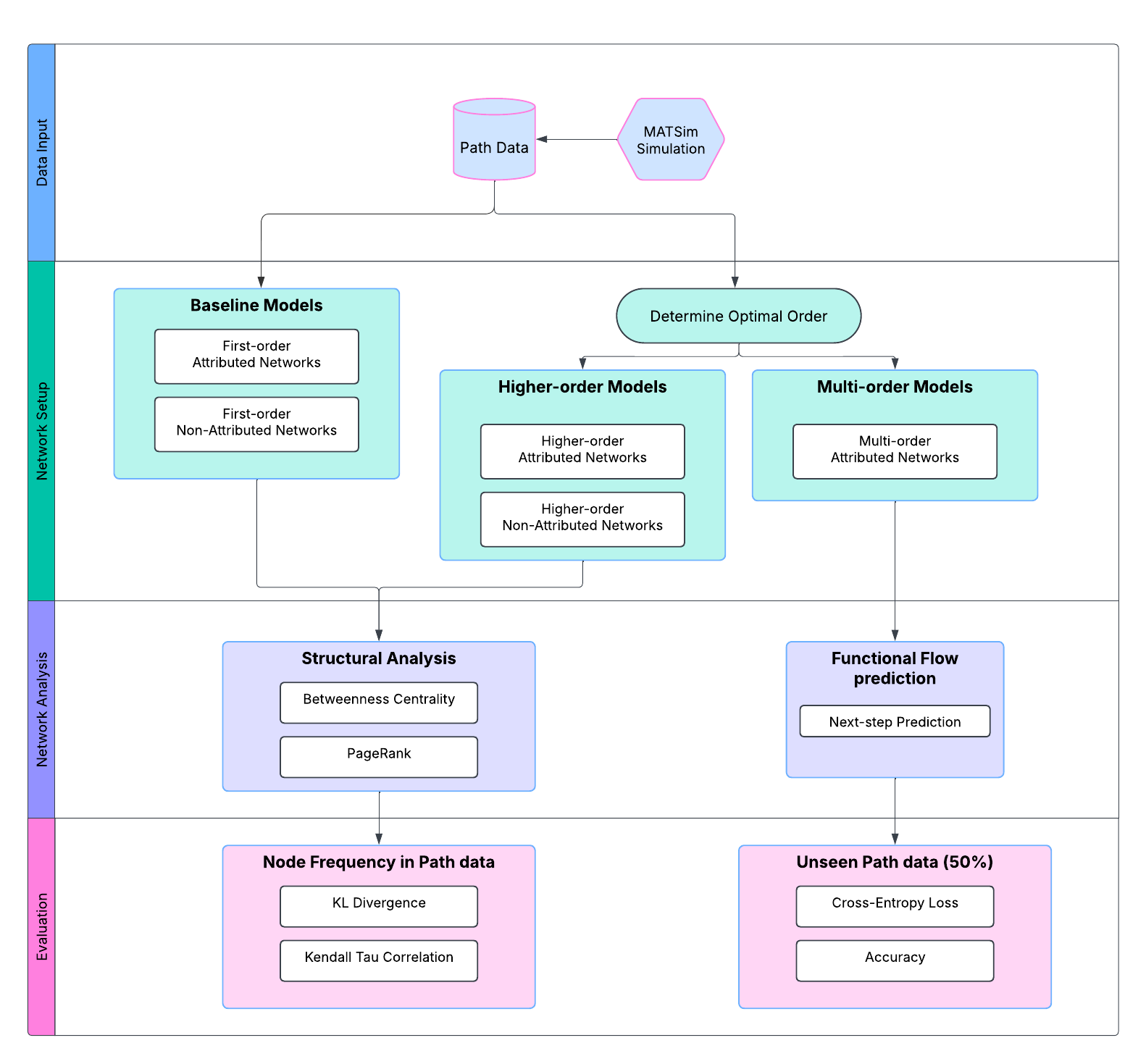}
    \caption{Experimental Workflow.
    A schematic overview of the modeling pipeline used to analyze the Sioux Falls transportation network. The framework addresses three research questions by constructing baseline (first-order attributed and non-attributed) and higher-order models. Higher-order models are built after determining the optimal order, using either transition probabilities (attributed) or path existence (non-attributed). All models are analyzed via Betweenness Centrality, PageRank, and next-step prediction (attributed only), and evaluated using KL divergence, Kendall’s tau, cross-entropy loss, and accuracy.}
    \label{fig:workflow}
\end{figure}
%
%
Model performance is evaluated against ground truth derived from the observed trajectory data.
For betweenness centrality, the ground truth corresponds to the frequency with which a node \(v\) (in first-order models) or a subpath \(v^{(k)}\) (in higher-order models) is traversed along paths between origin-destination pairs.
For PageRank, ground truth is based on node visitation frequencies aggregated over all simulated agent trajectories.
The quality of both centrality measures is assessed using two complementary metrics: Kullback-Leibler (KL) divergence \cite{Kullback1951Information}, which quantifies the discrepancy between the predicted and empirical frequency distributions, and Kendall's tau rank correlation \cite{Kendall1938New}, which evaluates the alignment between predicted and actual node rankings.
For the next-step prediction task, a set of 983,110 subpaths with known subsequent nodes is sampled from the trajectory data, using a 50\%/50\% train-test split. Predictive performance is evaluated using Cross-Entropy Loss to capture probabilistic accuracy, and Prediction Accuracy to measure correct next-step classification.

%
%
An overview of the experimental workflow is presented in Figure~\ref{fig:workflow}, while detailed descriptions of the evaluation metrics are provided in \ref{app:evaluation}.

\section{Results}\label{sec:results}
%
%
The results of the previously introduced analytical tasks are presented to demonstrate the applicability of higher-order network models in transportation research and to outline future directions for methodological and empirical advancements.
The proposed framework evaluates first-order and higher-order network models on the Sioux Falls benchmark using agent-based mobility data.
Both non-attributed and attributed empirically weighted variants are assessed across three tasks: betweenness centrality, PageRank, and next-step prediction.
Ground truth is derived from simulated trajectories, enabling quantitative comparison via rank correlation, divergence metrics, and predictive accuracy.

\paragraph{Optimal-Order.}
The optimal order is determined by applying consecutive likelihood ratio tests to multi-order models derived from pathway data, selecting the highest order \(k^{*}\) for which the improvement over the preceding model is statistically significant \cite{Scholtes2017When}.
For the Sioux Falls network, the analysis identifies an optimal order of \(k^{*} = 3\).
This suggests that incorporating a memory of the past three steps captures the dominant sequential dependencies in the mobility data.
Higher-order models (\(k > 3\)) introduce unnecessary model complexity, leading to overfitting by capturing noise and transitions that do not generalize.
In contrast, lower-order models (\(k < 3\)) fail to capture the temporal dependencies present in the observed trajectories, resulting in an underfitted model.
Consequently, we expect the third-order model to perform among the best in the subsequent evaluation tasks.

\paragraph{Higher-Order Betweenness Centrality.}
The evaluation of betweenness centrality demonstrates that higher-order network models outperform the first-order baseline, with the optimal third-order model showing the best overall performance (see Table~\ref{tab:bc-comparison}).
Kendall's tau correlation improve from 0.185 in the first-order network to 0.452 in the third-order network. 
KL divergence decrease from 1.597 in the first-order network to 0.034 in the third-order network. 
After the third order, performance slightly decrease or yield, with Kendall's tau at 0.440 and KL divergence at 0.039 for the fourth-order network. 
A similar trend is observed for non-attributed networks.
Kendall's tau correlation increases from 0.242 in the first-order network to 0.425 in the third-order network, with a marginal improvement at the fourth order (0.448) before decreasing at the fifth. 
KL divergence shows the lowest value at the third order (0.057), down from 0.367 in the first-order model, and increases slightly to 0.071 at the fourth order.

\begin{table}[htb]
  \centering
  \footnotesize
    \caption{Comparison of Betweenness Centrality by Orders}
    \label{tab:bc-comparison}
    \begin{tabular}{c|cc|cc}
    \toprule
    \multirow{2}{*}{\textbf{Order (k)}} & \multicolumn{4}{c}{\textbf{Betweenness Centrality}} \\ 
     & \multicolumn{2}{c|}{\textbf{KL Divergence}} & \multicolumn{2}{c}{\textbf{Kendall's Tau}} \\
     & Non-attributed & Attributed & Non-attributed & Attributed \\
    \midrule
    1 & 0.367 & 1.597 & 0.242 & 0.185  \\
    2 & 0.087 & 0.063 & 0.281 & 0.237  \\
    \textbf{3 (Optimal)} & \textbf{0.057} & \textbf{0.034} & 0.425 & \textbf{0.452} \\
    4 & 0.071 & 0.039 & \textbf{0.448} & 0.440  \\
    5 & 0.077 & 0.063 & 0.437 & 0.439 \\
    \bottomrule
    \end{tabular}
\end{table}

Higher-order networks generally outperform the first-order network, as illustrated in Figure~\ref{fig:bc-comparison}. 
For attributed networks, Kendall's tau correlation increase from 0.185 in the first-order network to 0.452 in the third-order network, and KL divergence decrease from 1.597 to 0.034. 
Non-attributed higher-order networks also show better performance than the first-order network. 
Kendall's tau correlation improved from 0.242 in the first-order network to 0.425 in the third-order network, and KL divergence decreased from 0.367 to 0.057.

\begin{figure}[htb]
    \centering
    \begin{subfigure}[h]{0.49\linewidth}
        \centering
        \includegraphics{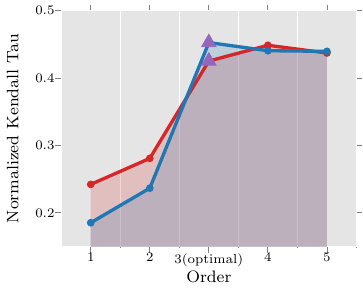}
        \caption{Kendall's tau correlation}
        \label{fig:bc-comparison-a}
    \end{subfigure}
    \hfill
    \begin{subfigure}[h]{0.49\linewidth}
        \centering
        \includegraphics{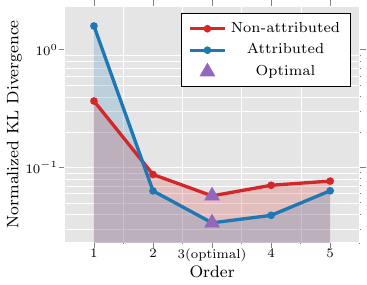}
        \caption{KL divergenc}
        \label{fig:bc-comparison-b}
    \end{subfigure}
    \caption{Comparison of betweenness centrality measures across different network orders (\(k=1\) to \(k = 5\)). (a) Kendall's tau correlation and (b) KL divergence simulated results and observed ground truth node frequencies. Higher-order models consistently outperform the first-order baseline, with significant improvement at the third-order.}
    \label{fig:bc-comparison}
\end{figure}

Attributed networks generally show higher correlations based on Kendall's tau. 
For example, Kendall's tau for attributed networks improve from 0.185 in the first-order network to 0.452 in the third-order network, while non-attributed networks improve from 0.242 to 0.425. 
However, attributed networks exhibit larger KL divergence in the first-order network (1.597) compared to the non-attributed network (0.367), suggesting that path data might be misleading for first-order networks. 
As the order increase, attributed networks showed more similar values compared to the ground truth, with the third-order network model achieving the best scores. 
Comparison through kernel density estimation (Figure~\ref{fig:bc-kde}) further confirms that higher-order models progressively align with the empirical node frequency distribution. In both attributed and non-attributed networks, the third-order model achieves the closest match.

\begin{figure}[h]
    \centering
    \begin{subfigure}[h]{\linewidth}
        \centering
        \includegraphics{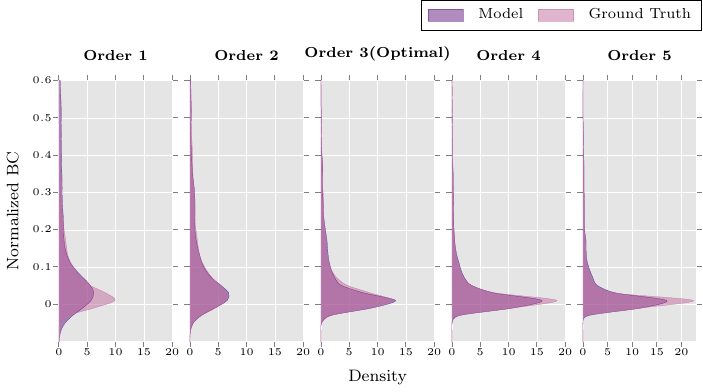}
        \caption{Attributed networks}
        \label{fig:bc-kde-a}
    \end{subfigure}
    \hfil
    \begin{subfigure}[h]{\linewidth}
        \centering
        \includegraphics{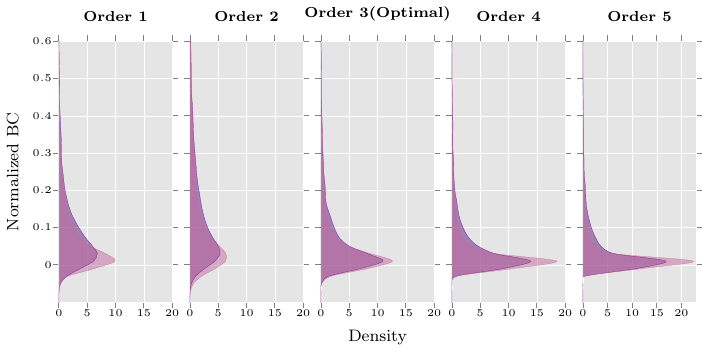}
        \caption{Non-attributed networks}
        \label{fig:bc-kde-b}
    \end{subfigure}
    \caption{Kernel Density Estimation (KDE) and frequency distribution plots of Betweenness Centrality. (a) Attributed networks and (b) Non-attributed networks across network orders (\(k=1\) to \(k = 5\)). Thereby higher-order models progressively align with the node frequency patterns in the observed path data.}
    \label{fig:bc-kde}
\end{figure}

\paragraph{Higher-Order PageRank.}
The third-order model ($k = 3$) demonstrates optimal performance for both attributed and non-attributed networks in terms of Kendall's tau correlation (see Table~\ref{tab:pr-comparison}).
For attributed networks, Kendall's tau correlation increase from 0.498 at first order to 0.617 at third order, and KL divergence decrease from 0.868 to 0.277. 
For non-attributed networks, Kendall's tau correlation similarly improved from 0.206 to 0.321, with KL divergence reducing from 1.541 to 0.907. 
Beyond this optimal order, Kendall's tau slightly decrease for both networks, whereas KL divergence continued to decrease for both attributed and non-attributed networks, suggesting continued alignment with ground truth distributions at higher orders.

\begin{table}[htb]
  \centering
  \footnotesize
    \caption{Comparison of PageRank by Orders}
    \label{tab:pr-comparison}
    \begin{tabular}{c|cc|cc}
    \toprule
    \multirow{2}{*}{\textbf{Order (k)}} &  \multicolumn{4}{c}{\textbf{PageRank}} \\ 
     & \multicolumn{2}{c|}{\textbf{KL Divergence}} & \multicolumn{2}{c}{\textbf{Kendall's Tau}}\\
     & Non-attributed & Attributed & Non-attributed & Attributed \\
    \midrule
    1 & 1.541 & 0.868 & 0.206 & 0.498  \\
    2 & 1.711 & 0.364 & 0.305 & 0.594  \\
    \textbf{3 (Optimal)} & 0.907 & 0.277 & \textbf{0.321} & \textbf{0.617} \\
    4 & 0.420 & 0.211 & 0.315 & 0.560 \\
    5 & \textbf{0.176} & \textbf{0.137} & 0.310 & 0.445 \\
    \bottomrule
    \end{tabular}
\end{table}

Higher-order networks consistently outperform the first-order network across all measures, as illustrated in Figure~\ref{fig:pr-comparison}. 
An exception is that the second-order model showed increased KL divergence (1.711) compared to the first-order model (1.541), suggesting that KL divergence alone might not reliably indicate optimal order in non-attributed scenarios.

\begin{figure}[htb]
    \centering
    \begin{subfigure}[h]{0.49\linewidth}
        \centering
        \includegraphics{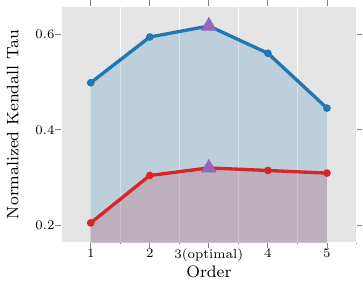}
        \label{fig:pr-comparison-a}
        \caption{Kendall's tau correlation}
    \end{subfigure}
    \hfill
    \begin{subfigure}[h]{0.49\linewidth}
        \centering
        \includegraphics{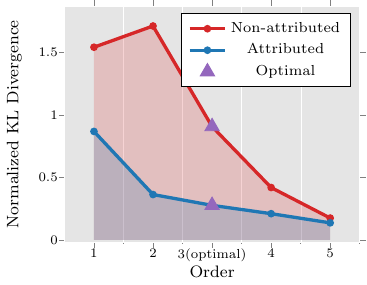}
        \caption{KL divergence}
        \label{fig:pr-comparison-b}
    \end{subfigure}
    \caption{Comparison of network centrality measures across different network orders ($k = 1-5$) illustrating: (a) Kendall's tau correlation and (b) KL divergence between PageRank (PR) values and ground truth node frequencies. Results demonstrate significant performance gains at third order, especially for attributed networks.}
    \label{fig:pr-comparison}
\end{figure}

Attributed networks generally show stronger Kendall's tau correlations compared to non-attributed networks, except in the first-order model (0.868 attributed vs. 1.541 non-attributed), indicating potential issues with path data representation at lower orders.
Kernel density estimates in Figure~\ref{fig:pr-kde} further support the quantitative results, revealing a progressive shift in PageRank distributions toward closer alignment with the empirical node visitation frequencies by increasing the order.

\begin{figure}[htb]
    \centering
    \begin{subfigure}[h]{\linewidth}
        \centering
        \includegraphics{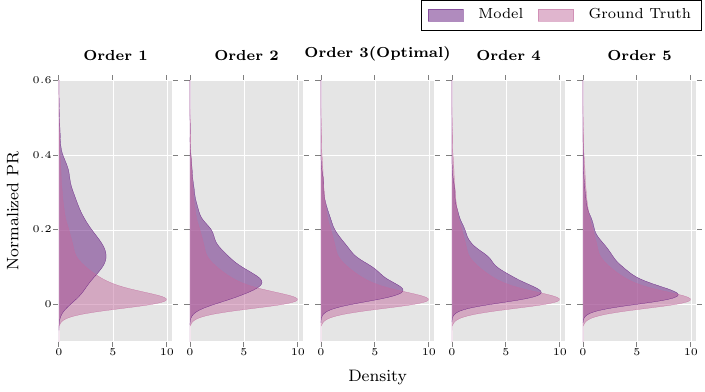}
        \caption{Attributed networks}
        \label{fig:pr-kde-a}
    \end{subfigure}
    \hfill
    \begin{subfigure}[h]{\linewidth}
        \centering
        \includegraphics{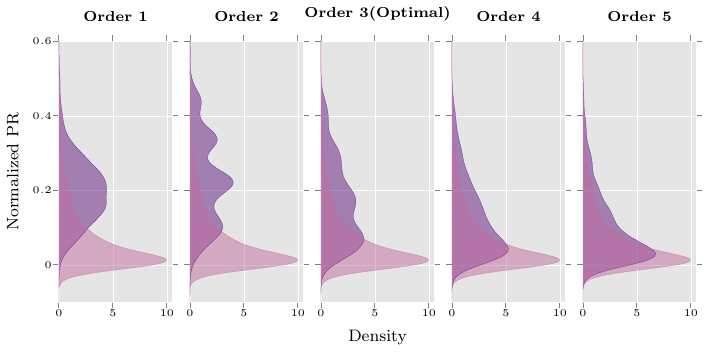}
        \caption{Non-attributed networks}
        \label{fig:pr-kde-b}
    \end{subfigure}
    \caption{Kernel Density Estimation (KDE) and frequency distribution plots of PageRank values. (a) Attributed networks and (b) Non-attributed networks across network orders (\(k=1\) to \(k = 5\)). Thereby the distributional shift in PageRank towards alignment with ground truth frequencies increases with higher network orders.}
    \label{fig:pr-kde}
\end{figure}

\paragraph{Next-Step Prediction.}
In addition to structural metrics, we validate the performance of the higher-order network models through next-step prediction experiments on the Sioux Falls network.
Thereby higher-order models improve prediction accuracy and reduce cross-entropy loss compared to the first-order model (see Table~\ref{tab:ns-comparison}). 
As shown in Figure~\ref{fig:ns-performance}, prediction accuracy increased from 62.1\% at first order to 87.5\% at third-order, eventually reaching 88.\% at fifth-order. 
Correspondingly, cross-entropy loss substantially decreased from 0.728 at first-order to 0.345 at second-order, further reducing to 0.313 at third-order. 
Beyond the third order, cross-entropy loss show an increase to 0.315 at fourth-order, followed by 0.338 at fifth-order, suggesting potential overfitting at higher network orders. 
These results may indicate that the third-order model achieves an balance between predictive accuracy and model complexity.

\begin{figure}[htb]
    \centering
    \includegraphics{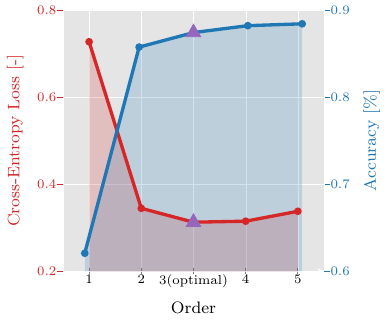}
    \caption{ Next-step prediction accuracy for the Sioux Falls network using higher-order network models. Prediction accuracy improves from 62.1\% at first-order to a peak of 88.5\% at fifth-order, with the third-order model providing a balance between cross-entropy loss and accuracy.}
    \label{fig:ns-performance}
\end{figure}

\begin{table}[htb]
  \centering
  \footnotesize
  \caption{Comparison of Next-Step Prediction (Attributed)}
  \label{tab:ns-comparison}
  \begin{tabular}{c|cc}
    \toprule
    \multirow{2}{*}{\textbf{Order (k)}} & \multicolumn{2}{c}{\textbf{Prediction}} \\ 
     &  \textbf{Los [-]} & \textbf{Accuracy [\%])} \\
    \midrule
    1 & 0.728 & 62.1 \\
    2 & 0.345 & 85.8 \\
    \textbf{3 (Optimal)} & \textbf{0.313} & 87.5 \\
    4 &  0.315 & 88.2 \\
    5 &  0.338 & \textbf{88.5} \\
    \bottomrule
    \end{tabular}
\end{table}

\paragraph{Computational Complexity.}

All experiments are conducted on a workstation equipped with an Intel\textregistered\ Xeon\textregistered\ w5-2465X CPU (32 logical cores, up to 4.7~GHz) and an NVIDIA RTX A2000 12GB GPU.
The CPU supports advanced SIMD instructions and virtualization (VT-x), enabling efficient execution of computationally intensive graph algorithms.
The GPU, running driver version 535.183.01 with CUDA 12.2, is utilized for accelerating selected tasks that benefit from parallelism.
Table~\ref{tab:computational_expense} summarizes the runtime performance for key graph analysis tasks across different network orders.
For betweenness centrality, we report two variants: the attributed implementation computed on the CPU using \texttt{NetworkX} \cite{Hagberg2008Exploring}, and the non-attributed implementation executed on the GPU using \texttt{cuGraph} \cite{NVIDIA2025cuGraph}.
For PageRank, only CPU runtimes are reported.

\begin{table}[h]
  \centering
  \footnotesize
    \caption{Computational expense (in seconds).}
    \label{tab:computational_expense}
    \begin{tabular}{c|cc|cc}
        \toprule
        \multirow{3}{*}{\textbf{Order (k)}} & \multicolumn{2}{c|}{\textbf{Betweenness Centrality}} & \multicolumn{2}{c}{\textbf{PageRank}} \\
      \cmidrule(lr){2-3} \cmidrule(lr){4-5}
                                            & Non-attributed & Attributed & Non-attributed & Attributed \\
                                            & GPU & CPU & CPU & CPU \\

        \midrule
        1  & 20.3766  & 4.7573   & 0.00324 & 0.00669 \\
        2  & 43.3353  & 25.0310  & 0.00517 & 0.01216 \\
        3  & 69.7954  & 80.2084  & 0.01161 & 0.02332 \\
        4  & 101.3077 & 190.1924 & 0.02874 & 0.04338 \\
        5  & 134.9482 & 351.2743 & 0.09691 & 0.09511 \\
        \bottomrule
    \end{tabular}
  \end{table}

\paragraph{Implementation.}  
The construction and analysis of higher-order network models in this study are implemented using the open-source Python library \texttt{pathpyG} \cite{Hackl2021Analysis}, which provides a unified framework for GPU-accelerated higher-order network analytics and learning.  
Specifically designed for processing sequential and time series data on networks, \texttt{pathpyG} supports the efficient construction of de Bruijn-based network representations, encoding of memory-dependent transition structures, and scalable estimation of higher-order Network models.  
The code used to construct the higher-order network models, perform the experiments, and reproduce all results presented in this manuscript is available on GitHub \footnote{GitHub: \url{https://github.com/cisgroup/hon-transportation}. The code will be archived and made permanently available via Zenodo upon acceptance of the manuscript for publication.}.

\section{Discussion}\label{sec:discussion}
%
%
The results presented in this study demonstrate that higher-order network models outperform traditional first-order approaches in accurately capturing structural and dynamic properties of transportation systems.
Across all evaluation tasks, including betweenness centrality, PageRank, and next-step predictions, higher-order representations consistently showed improved alignment with empirical ground truth derived from agent-based mobility data.
Notably, the third-order model achieved optimal performance, indicating that a memory of three previous steps captures the dominant sequential dependencies observed in real-world-like routing behavior.
This suggests that incorporating memory effects into network models leads to a more realistic and data-consistent characterization of mobility dynamics, thereby enhancing the explanatory and predictive power of network-based transportation analyses.

%
By extending classical graph-based methodologies through the integration of memory-dependent dynamics, this work contributes a novel analytical framework to the field of transportation science.
The use of higher-order network models based on Markov chains and de Bruijn graphs allows for the development of centrality measures and flow models that reflect not only topological structure but also the temporally ordered sequences of movement inherent in mobility data.
As a result, traditional concepts such as betweenness centrality and PageRank can be reinterpreted in the context of memory-aware agent movement, enabling the identification of structurally and functionally critical network components with improved fidelity.
In real-world applications where detailed trajectory data is not available, non-attributed models can be used to capture the underlying movement patterns without relying on edge frequency information, offering a lightweight yet effective alternative for data-scarce environments.
Conversely, when empirical path data is available, the framework enables a fully data-driven approach by directly inferring higher-order models through maximum likelihood estimation.
Alternatively, when trajectory data is available, the framework supports a fully data-driven approach by enabling the direct inference of higher-order models through maximum likelihood estimation.
This allows for the construction of memory-aware network representations grounded in observed mobility patterns.
Furthermore, the integration of a statistically rigorous model selection strategy based on likelihood-ratio tests facilitates the adaptive determination of the optimal memory order, thereby avoiding arbitrary choices of model complexity and enhancing model interpretability and generalization.
Taken together, these innovations provide a more generalizable and flexible modeling framework for studying complex transportation systems across a range of spatial and temporal scales, accommodating both rich trajectory datasets and more abstract, data-scarce topology-driven scenarios.

%
The framework developed in this study is highly generalizable and applicable across a broad range of transportation scenarios.
Although evaluated on the Sioux Falls benchmark network, the methodology is compatible with real-world datasets derived from GPS traces, mobile phone records, or public transit smart card data.
Moreover, the approach is readily extendable to multimodal transportation systems, where each mode, such as walking, biking, or transit, can be modeled within its own memory-aware subnetwork, or integrated into a multilayer higher-order model.
Beyond transportation, the concepts presented here have potential applications in other domains involving spatial or temporal pathways, such as supply chain logistics, pedestrian dynamics in buildings, or power flows in energy networks.
In a broader context, the higher-order representations introduced in this work can also serve as an enriched input for graph learning tasks, enabling data-driven models to capture complex dependencies that extend beyond simple pairwise interactions \cite{Qarkaxhija2022Bruijn, Heeg2023Using}.

%
Despite the demonstrated advantages, the proposed framework faces several limitations that warrant careful consideration.
One primary challenge is the scalability of higher-order models, as the number of nodes and edges in the de Bruijn-based representation increases combinatorially with the memory order, leading to increased computational and memory overhead for large-scale networks or high-order representations.
Furthermore, while attributed models that incorporate empirical traversal frequencies improve accuracy, they also introduce the risk of overfitting when data is sparse or unevenly distributed over the system.
The use of synthetic trajectory data, although widely accepted in benchmarking, may not fully capture the diversity of real-world travel behavior, potentially limiting generalizability.
Additionally, higher-order models, while powerful, are less interpretable than their first-order counterparts, especially in terms of traditional graph-theoretic intuition and visualization.
The problem of projecting results from higher-order network structures onto first-order representations, in a manner that retains key dynamic information while enhancing interpretability and communicability, remains an open and challenging research question.

%
Building upon the findings of this study, future research should focus on validating the proposed higher-order network framework using real-world trajectory data from diverse transportation systems.
This would enable a rigorous assessment of the generalizability and practical applicability of memory-aware models in capturing complex mobility patterns and traffic dynamics beyond synthetic benchmarks.
In particular, real-world validation would support the translation of higher-order modeling insights into tangible improvements in transportation planning, congestion mitigation, and infrastructure design.
Another important direction lies in the development of alternative mapping strategies between higher-order and first-order network representations.
While this study relies on projecting higher-order node scores to their terminal components, exploring alternative aggregation schemes, such as weighted averaging or learned mappings, could lead to better alignment with empirical flow patterns and improved performance in task-specific applications, such as routing optimization or vulnerability analysis.
Furthermore, future research could explore methods for inferring higher-order network structures directly from first-order topologies or partially observed data.
Such inference techniques would address challenges associated with data sparsity, enabling the construction of reliable models even in scenarios where complete trajectory datasets are unavailable.
Leveraging non-attributed network structures as informative priors for parameter estimation offers a particularly promising strategy for recovering transition dynamics and estimating critical probabilities that govern agent movement.
In parallel, integrating higher-order network models with advanced machine learning techniques opens up powerful new possibilities.
Graph learning methods such as Variational Graph Auto-Encoders \cite{Kipf2016Variational} in combination with
Transformer \cite{Vaswani2017Attention} based path encoding, could be trained on higher-order representations to uncover latent mobility patterns, learn optimal memory depths, or infer behavioral rules from data.
This integration would enable the automated adaptation of network models to specific urban contexts, improving predictive capabilities and real-time decision-making in transportation systems.
%
%
Finally, advancing the theoretical foundations of higher-order network dynamics remains an important open challenge.
A deeper mathematical understanding of how memory structures influence emergent flow behavior could reveal fundamental principles of transport efficiency, resilience, and optimal control.
For example, studying the structural properties of non-attributed higher-order models may provide insights into naturally emerging optimal transport pathways or invariant flow structures, thus offering a new perspective through which to understand the performance and organization of complex infrastructure systems.

\section{Conclusions}\label{sec:conclusions}
%
%
This study introduced a higher-order network modeling framework for transportation systems that captures memory-dependent dynamics inherent in empirical mobility trajectories.
By extending classical graph-based representations through higher-order Markov chains and de Bruijn-based network structures, the proposed approach integrates spatial-temporal dependencies into both structural and dynamic network analyses.
%
%
The empirical evaluation on the Sioux Falls benchmark network demonstrated that incorporating memory effects substantially improves the accuracy and fidelity of network analytic methods.
Higher-order extensions of centrality measures, such as betweenness centrality and PageRank, yielded rankings that more closely aligned with observed flow distributions and visitation frequencies.
In addition, the predictive accuracy of next-step forecasting tasks improved significantly with increasing model order, with the third-order model achieving an optimal trade-off between model complexity and predictive performance.
%
%
These results highlight the limitations of traditional first-order models in capturing the sequential structure of transportation flows and underscore the importance of memory-aware representations for more realistic modeling of agent behavior.
The proposed higher-order modeling framework enables the adaptive selection of memory length, supporting robust and data-driven construction of network representations that generalize beyond specific datasets or domains.
%
%
Overall, this work establishes a methodological foundation for the application of higher-order network models in transportation science and offers a scalable and generalizable framework for analyzing and predicting mobility dynamics.
Future research will focus on applying this approach to real-world trajectory data, integrating it with graph learning architectures, and exploring its theoretical implications for understanding and optimizing transport systems under uncertainty and dynamical complexity.

\section*{ORCID}

\noindent Chen Zhang - \url{https://orcid.org/0009-0002-7537-8414}

\noindent Jürgen Hackl - \url{https://orcid.org/0000-0002-8849-5751}

\newpage
\bibliographystyle{ws-acs}
\bibliography{bibliography}

\newpage
\appendix
\section{Nomenclature}
\begin{table}[h!]
\centering
\resizebox{\textwidth}{!}{
\begin{tabular}{|l|l|}
\hline
\multicolumn{2}{|c|}{\textbf{Network Notation}} \\
\hline
$\bm{G}$ & Graph or network \\
$\mathcal{V}$ & Set of nodes (vertices) \\
$\mathcal{E}$ & Set of edges (links) \\
$(v,w)$ & Directed edge from node $v$ to node $w$ \\
$\mathrm{d}(v,w)$ & Distance between nodes $v$ and $w$ \\
$\mathrm{w}(v_{i}, v_{i+1})$ & Weight of the edge between nodes $v_i$ and $v_{i+1}$ \\
\hline
\multicolumn{2}{|c|}{\textbf{Path and Markov Chain}} \\
\hline
$p$ & Path in a network \\
$v_i$ & Node in the path $p$ \\
$X_t$ & State of the walker at time $t$ \\
$\Omega$ & State space of the random walk \\
$P ( X_{t+1} = s_{t+1} \mid X_t=s_{t} )$ & Transition probability in a random walk \\
$\bm{T}$ & Transition matrix \\
$\bm{A}$ & Adjacency matrix \\
$k$ & Order of the Markov chain \\
$P^{(k)}(s_t \mid s_{t-1}, \ldots, s_{t-k})$ & Transition probability in a $k$th-order Markov chain \\
\hline
\multicolumn{2}{|c|}{\textbf{Higher-Order Network}} \\
\hline
$v^{(k)}, w^{(k)}$ & Higher-order nodes (sequence of $k$ consecutive nodes) \\
$L(G)$ & Line graph of the original network $G$ \\
$L^2(G)$ & Second-order line graph of $G$ ($L(L(G))$) \\
$\text{w}\bigl(\langle v_{1}, \ldots, v_{k} \rangle,\ \langle v_{2}, \ldots, v_{k+1} \rangle\bigr)$ & Weight of a higher-order edge connecting consecutive node sequences \\
\hline
\multicolumn{2}{|c|}{\textbf{Multi-Order Model and Likelihood}} \\
\hline
$\bar{P}^{(K)}(p)$ & Likelihood of path $p$ under a multi-order Markov model \\
$\mathcal{L}(M_k \mid p)$ & Likelihood of observing a path $p$ under model $M_k$ \\
$\mathcal{P}$ & Multi-set of statistically independent paths \\
$\widehat{P}^{(k)}(v_i \mid v_{i-1}, \ldots, v_{i-k})$ & Estimated transition probability in a $k$th-order Markov chain \\
\hline
\multicolumn{2}{|c|}{\textbf{Betweenness Centrality}} \\
\hline
$C_B(v)$ & Betweenness centrality of a node $v$ \\
$\sigma_{st}$ & Total number of shortest paths between nodes $s$ and $t$ \\
$\sigma_{st}(v)$ & Number of shortest paths between $s$ and $t$ passing through $v$ \\
$C_B^{(k)}(v)$ & Higher-order betweenness centrality of a node $v$ \\
$\sigma_{st}^{(k)}$ & Total number of shortest paths in the $k$th-order network \\
$\sigma_{st}^{(k)}(v)$ & Number of shortest paths in the $k$th-order network passing through $v$ \\
\hline
\multicolumn{2}{|c|}{\textbf{PageRank}} \\
\hline
$\bm{r}$ & PageRank vector \\
$\alpha$ & Teleportation (damping) factor \\
$\bm{E}$ & Matrix of all ones \\
$\bm{r}^{(k)}$ & Higher-order PageRank vector \\
$\tilde{r}(v)$ & Aggregated first-order PageRank score for node $v$ \\
\hline
\end{tabular}
}
\label{table:notation}
\end{table}

\newpage
\section{Evaluation Methods}\label{app:evaluation}

\subsection{Kendall’s Rank Correlation \cite{Kendall1938New}}

Kendall’s rank correlation coefficient, denoted by \( \tau \), is a non-parametric test that measures the ordinal association between two measured quantities. This study compares the betweenness centrality produced by higher-order network models with ground truth data based on node frequency. The Kendall’s \( \tau \) coefficient is calculated as:

\[
\tau = \frac{2(C - D)}{n(n-1)}
\]

where \( C \) represents the number of concordant pairs, which are pairs where the ranking order of both elements agrees in both datasets, and \( D \) denotes the number of discordant pairs, where the ranking order of elements disagrees between the datasets. The variable \( n \) is the total number of pairs being compared.

Kendall’s \( \tau \) ranges from $-1$, indicating complete disagreement, to $+1$, indicating complete agreement, with 0 representing no correlation. This metric is particularly useful for validating the rank order of nodes as determined by higher-order models against the observed ground truth.

\subsection{Kullback-Leibler Divergence \cite{Kullback1951Information}}
The Kullback-Leibler (KL) divergence, also known as relative entropy, is a measure of how one probability distribution diverges from a second expected probability distribution. It is defined as:
\[
D_{KL}(P || Q) = \sum_{i} P(i) \log\left(\frac{P(i)}{Q(i)}\right)
\]
where \( P \) and \( Q \) are two probability distributions over the same variable. The KL divergence is always non-negative and is equal to zero if and only if the two distributions are identical. It is particularly useful in information theory and statistics for measuring the difference between two probability distributions.

\section{Optimal Order Detection for Multi-order Networks\cite{Scholtes2017When}}
\label{app:optimal-order}
Multi-order networks integrate multiple Markov chain orders into a single framework. In a multi-order model with maximum order \(k\), the probability of a path 
\[
p = \left(v_1, v_2, \ldots, v_{\ell}\right)
\]
is expressed as
\begin{equation}
    \begin{aligned}
    \bar{P}^{(k)}(p)=\;& \prod_{i=1}^{k} P^{(i)}\left(v_i \mid v_1, \ldots, v_{i-1}\right) \\
    &\times \prod_{j=k+1}^{\ell} P^{(k)}\left(v_j \mid v_{j-k}, \ldots, v_{j-1}\right).
    \end{aligned}
\end{equation}
The first product captures transitions using models of increasing order (from 1 up to \(k\)), while the second product applies the highest-order model for transitions beyond the initial \(k\) states. This multi-layered approach enables the model to flexibly capture both short-term and long-term dependencies in the data.

Determining the appropriate memory length in transportation networks is critical because these networks exhibit dependencies at multiple scales; some transitions depend on only the immediate past, while others are influenced by a longer history. To detect the optimal order \(k_{\text{opt}}\), we model the probability of each observed path using the multi-order model. For a set of \(N\) paths \(\mathcal{P}=\{p_j\}_{j=1}^N\), the likelihood under the \(k\)-th order model is given by
\[
\mathcal{L}\left(\bar{M}_k \mid \mathcal{P}\right)=\prod_{j=1}^N \bar{P}^{(k)}\left(p_j\right).
\]
However, simply maximizing this likelihood (via maximum likelihood estimation) may lead to overfitting, as higher-order models contain more parameters. To balance model complexity and explanatory power, we compare nested multi-order models \(\bar{M}_k\) (null model) and \(\bar{M}_{k+1}\) (alternative model) using a Likelihood Ratio Test (LRT). The null hypothesis is that \(\bar{M}_k\) sufficiently explains the data, while the alternative hypothesis posits that \(\bar{M}_{k+1}\) provides a significantly better fit. We compute the likelihood ratio statistic
\begin{equation}
\Lambda=-2 \ln \frac{\mathcal{L}\left(\bar{M}_k \mid \mathcal{P}\right)}{\mathcal{L}\left(\bar{M}_{k+1} \mid \mathcal{P}\right)},
\end{equation}
which, under Wilks’ theorem, asymptotically follows a \(\chi^2\)-distribution with degrees of freedom \(\Delta d = d(k+1) - d(k)\). Here, the degrees of freedom are adjusted to account for feasible paths in the graph:
\begin{equation}
d(k)=(|\mathcal{V}|-1)+\sum_{i=1}^k\Biggl[\#\text{ paths of length } i - \#\text{ non-zero rows in } \bm{A}^i\Biggr],
\end{equation}
where \(\bm{A}^i\) is the \(i\)-th power of the adjacency matrix and the non-zero rows are counted via the Heaviside step function.

We then compute the p-value as
\begin{equation}
p=1-\frac{\gamma\left(\frac{\Delta d}{2},\frac{\Lambda}{2}\right)}{\Gamma\left(\frac{\Delta d}{2}\right)},
\end{equation}
where \(\gamma(s,x)\) is the lower incomplete Gamma function and \(\Gamma(s)\) is the complete Gamma function. By systematically comparing models of increasing order and selecting the largest \(k\) for which \(p<\epsilon\) (e.g., \(\epsilon=0.05\)), we obtain an optimal memory length \(k_{\text{opt}}\) that avoids both underfitting and overfitting. This data-driven approach ensures a balanced and robust model for capturing the dynamics in transportation networks.

\end{document}